\documentclass[a4paper, 11pt]{article}

\usepackage{amsmath}
\usepackage{amssymb}
\usepackage{amsthm}
\usepackage{amsfonts}
\usepackage{mathtools}
\usepackage{ascmac}
\usepackage{spalign}
\usepackage{fancybox} 
\usepackage[mathscr]{eucal}
\usepackage{footnpag}
\usepackage{siunitx}
\usepackage{multicol}
\usepackage{bbm}
\usepackage{framed}
\usepackage{here}
\usepackage{braket}

\theoremstyle{definition}
\newtheorem{theorem}{Theorem}[section]
\newtheorem*{theorem*}{Theorem}
\newtheorem{definition}{Definition}[section]
\newtheorem*{definition*}{Definition}

\newtheorem*{proposition*}{Proposition}
\newtheorem{lemma}{Lemma}[section]
\newtheorem*{lemma*}{Lemma}
\newtheorem{corollary}{Corollary}[section]
\newtheorem*{corollary*}{Corollary}

\newtheorem*{remark*}{Remark}

\title{Hadamard Rigidity of Positive Sojourn Time Distributions for Rotation Coins}
\author{
Shunya Tamura\thanks{
Corresponding author.
Okegawa City Okegawa West Junior High School,
Saitama, 363-0027, Japan,
e-mail: shunya.tamura059@gmail.com
}
\and
Tomoki Yamagami\thanks
{Department of Information and Computer Sciences, 
Saitama University, Saitama, 338-8570, Japan, 
e-mail: tyamagami@mail.saitama-u.ac.jp
}
}

\date{}

\begin{document}
\maketitle
\begin{abstract}
We study the distribution of the positive sojourn time
for a one-dimensional two state quantum walk,
conditioned on return to the origin.
Konno showed that, for the Hadamard walk,
this conditional distribution is exactly uniform
at times divisible by $4$.
In this paper, we investigate whether this finite time exact uniformity
characterizes the Hadamard coin within the family of rotation coins.

For a fixed initial state, we prove that the following three conditions
are equivalent for rotation coins:
the conditional distribution is exactly uniform at time $8$;
the conditional distribution is exactly uniform at every time $4m$
with $m\ge2$;
and the coin is the Hadamard coin.
Thus, the uniformity phenomenon found by Konno
is characterized as a rigidity phenomenon of the Hadamard coin
within the rotation coin family.

The proof uses a matrix-valued generating function
for paths returning to the origin.
We analyze the algebraic structure arising from
an absorbing process on the half line.
Finally, by comparing low degree coefficients at time $8$,
we show that exact uniformity forces the rotation coin
to be the Hadamard coin.
\end{abstract}



\section{Introduction}

A one-dimensional quantum walk with two internal states
is a fundamental model at the intersection of
quantum information, probability theory, and combinatorics \cite{Ambainis2001, Konno2002, Kempe2003, VenegasAndraca2012}.
It exhibits distributional structures which are essentially different
from those of classical random walks.
In a classical random walk, the probability distribution is obtained by summing the probabilities of all paths leading to each position.
For a quantum walk, in contrast,
complex amplitudes are superposed first,
and the probability of finding the walker at a position is then obtained by taking the squared norm of the resulting amplitude.
This quantum interference gives rise to characteristic features including ballistic spreading and distinctive finite-time phenomena.
Basic properties of one-dimensional quantum walks
with two internal states have been studied in many works
after Ambainis et al. \cite{Ambainis2001}, including
limit theorems for the probability distribution of one-dimensional quantum walks \cite{Konno2002, Konno2005}.

Quantum walks also have interesting structures
in their distributions at finite times,
in addition to their limit behavior as time tends to infinity.
One such phenomenon is the exact uniform distribution
of the sojourn time in the positive half-line
for the one-dimensional quantum walk with the Hadamard coin,
conditioned on return to the origin,
as established by Konno \cite{Konno2012}.
Here the sojourn time in the positive half-line is defined
according to Konno's convention:
a time interval at the origin is assigned to the positive
or negative side according to the next direction of motion.
Under this convention, in the case of the Hadamard coin,
the distribution of the sojourn time in the positive half-line,
conditioned on return to the origin,
has a very symmetric form.

In subsequent work, the sojourn time for one-dimensional quantum walks
with general coin operators has also been studied
\cite{Cai2023}.
These studies treat distributions and computational methods
for the sojourn time under general coins,
and they extend the study of sojourn times in quantum walks.
The problem considered in the present paper is different.
Our purpose is not to determine the sojourn time distribution
for a general coin.
Rather, we investigate whether the exact uniform distribution
at finite times, which appears in Konno's Hadamard walk,
is rigid enough to characterize the coin matrix.

In this paper, we consider the rotation coin
\[
U(\theta)=
\begin{bmatrix}
\cos\theta & \sin\theta\\
\sin\theta & -\cos\theta
\end{bmatrix}
\qquad
\left(0<\theta<\frac{\pi}{2}\right)
\]
and the initial state
\[
\varphi_\ast=
\frac{1}{\sqrt2}
\begin{bmatrix}
1\\ i
\end{bmatrix}.
\]
We analyze the distribution of the sojourn time
in the positive half-line, conditioned on return to the origin.
Our main result shows that, in this family of rotation coins,
Konno's exact uniform distribution occurs only for the Hadamard coin.
More precisely, we prove that the following three conditions
are equivalent:
the distribution of the sojourn time in the positive half-line
conditioned on return to the origin is exactly uniform at time $8$;
the same distribution is exactly uniform at time $4m$
for every $m\ge2$;
and
\[
\theta=\frac{\pi}{4}.
\]
Thus, the uniformity at time $8$ already forces the coin
to be the Hadamard coin.
In this sense, the uniform distribution phenomenon at finite times
discovered by Konno is characterized as a rigidity phenomenon
specific to the Hadamard coin in the family of rotation coins.

In the proof, we introduce a generating function with matrix entries
for paths returning to the origin,
and study its algebraic structure.
In the case of rotation coins, the absorbing process on the half line
leads to the square root
\[
\sqrt{z^4t^4+2\cos(2\theta)z^2t^2+1}.
\]
When $\theta=\pi/4$, which corresponds to the Hadamard coin, we have
\[
\cos(2\theta)=0,
\]
and this square root becomes simpler.
However, the necessity of the exact uniform distribution
does not follow only from the form of this square root.
Therefore, in this paper, we compare coefficients of low degree
at time $8$ and show that the exact uniformity forces
\[
\cos(2\theta)=0.
\]

The remainder of this paper is organized as follows.
In Section \ref{sec:definition},
we define one-dimensional quantum walks with two internal states
and the sojourn time in the positive half-line.
In Section \ref{sec:uniform},
we introduce a generating function with matrix entries
for paths returning to the origin,
and prove the characterization theorem for the exact uniform distribution
under rotation coins.
In Section \ref{sec:numerical},
we give numerical examples for small times
and confirm that the uniform distribution appears only
for the Hadamard coin.
Finally, we give the conclusion and some future problems.

\section{Definition of the quantum walk and the positive sojourn time}
\label{sec:definition}

In this section, we give the definition of a one-dimensional discrete-time quantum walk
with two internal states,
and define the positive sojourn time used in this paper.
In particular, following the definition introduced by Konno \cite{Konno2012},
we adopt the convention that an interval at the origin is assigned
to the positive side or to the negative side according to the next direction
of motion.

We first define a one-dimensional discrete-time quantum walk with two internal states.
Let $\mathbb Z$ be the set of all integers.
The quantum walker has two chiralities,
the left chirality $\ket{L}$ and the right chirality $\ket{R}$,
which are represented by
\[
\ket{L}=
\begin{bmatrix}
1\\0
\end{bmatrix},
\qquad
\ket{R}=
\begin{bmatrix}
0\\1
\end{bmatrix}.
\]
The time evolution on the coin space is given by a $2\times2$ unitary matrix
\[
U=
\begin{bmatrix}
a & b\\
c & d
\end{bmatrix}
\qquad
(a,b,c,d\in\mathbb C).
\]
One step of the time evolution of the quantum walker is defined as follows.
First, the coin matrix $U$ acts on the coin state.
Then the walker moves to the left or to the right according to its chirality.
For this purpose, we decompose $U$ as
\[
P=
\begin{bmatrix}
a & b\\
0 & 0
\end{bmatrix},
\qquad
Q=
\begin{bmatrix}
0 & 0\\
c & d
\end{bmatrix}.
\]
In this paper, $P$ corresponds to a move to the left,
and $Q$ corresponds to a move to the right.
We also use the convention that matrix products act from right to left.
Thus, for example, $QP$ represents a path which first moves to the left
and then moves to the right.

For time $n = 1,\,2,\,\dots$, let $\Xi_n(\ell,m)$ denote the sum of the matrix products
corresponding to all paths which move $\ell$ times to the left
and $m$ times to the right.
Here $n=\ell+m$.
For example,
\begin{align*}
\Xi_2(1,1)&=QP+PQ,\\
\Xi_4(2,2)&=Q^2P^2+P^2Q^2+QPQP+PQPQ+PQ^2P+QP^2Q
\end{align*}
hold.

Let $\varphi_\ast\in\mathbb C^2$ be the initial coin state,
and assume that
\[
\|\varphi_\ast\|=1.
\]
When the walker starts from the origin,
the probability that the walker is at position $x\in\mathbb Z$ at time $n = 1,\,2,\,\dots$ is given by
\[
\mathbb P(X_n=x)
=
\left\|
\Xi_n\left(\frac{n-x}{2},\frac{n+x}{2}\right)
\varphi_\ast
\right\|^2
\]
if $n+x$ is even.
Here
\[
\ell=\frac{n-x}{2},
\qquad
m=\frac{n+x}{2},
\]
and hence $x=-\ell+m$.
On the other hand, if $n+x$ is odd, then
\[
\mathbb P(X_n=x)=0.
\]

We next define the positive sojourn time used in this paper.
This positive sojourn time does not simply count
the number of times at which the position is positive.
Let $X_j$ be the position of the walker at time $j$.
We define the walker to be on the positive side during the time interval $[j,\,j+1)$ if the following condition holds:
\[
X_j>0,
\qquad\text{or}\qquad
X_j=0\ \text{and}\ X_{j+1}=1.
\]
In all other cases, the walker is regarded as being on the negative side during $[j,\,j+1)$.
Note that, if $X_j=0$ and the walker next moves to the left,
then the interval $[j,\,j+1)$ is assigned to the negative side.
Under this convention, the positive sojourn time at time $n$
is defined as the number of time intervals in $[0,\,n)$
during which the walker is on the positive side.

Let $\Psi_n^x(k)$ be the sum of the matrix products
corresponding to all paths which start from the origin,
arrive at position $x$ at time $n$,
and have positive sojourn time $k$.
In particular, for paths returning to the origin, we write
\[
\Gamma_n(k):=\Psi_n^0(k).
\]
For example, the following identities hold:
\begin{align*}
\Gamma_2(0)&=QP, & \Gamma_2(1)&=O, & \Gamma_2(2)&=PQ,\\
\Gamma_4(0)&=Q^2P^2+QPQP, & \Gamma_4(1)&=O,
&\Gamma_4(2)&=QP^2Q+PQ^2P, \\
\Gamma_4(3)&=O, & \Gamma_4(4)&=P^2Q^2+PQPQ.
\end{align*}
Here $O$ is the $2\times2$ zero matrix.

For the initial state $\varphi_\ast$, define
\[
w_n(k):=\|\Gamma_n(k)\varphi_\ast\|^2
\qquad
(0\le k\le n).
\]
The quantity $w_n(k)$ is the unnormalized quantum probability weight
of paths that return to the origin at time $n$
and have positive sojourn time $k$.
Therefore, the distribution of the positive sojourn time
conditioned on return to the origin is given by
\[
\mu_n(k)
:=
\frac{w_n(k)}
{\sum_{j=0}^{n}w_n(j)}
=
\frac{\|\Gamma_n(k)\varphi_\ast\|^2}
{\sum_{j=0}^{n}\|\Gamma_n(j)\varphi_\ast\|^2}
\qquad
(0\le k\le n).
\]
Here we restrict attention to times $n$ for which the denominator is positive.
If the denominator is zero, then the conditional distribution $\mu_n(k)$
is not defined.

By adopting this convention,
the positive sojourn time distribution with return to the origin
exhibits a natural symmetry.
It is also compatible with the simplification of the generating function
and with the exact uniform distribution at finite times shown later.
If one uses instead the definition which simply counts the times
at which $X_j>0$,
then the exact uniform distribution at finite times studied in this paper
does not hold in general.

\section{Simplification of generating functions and exact uniform distributions}
\label{sec:uniform}

In this section, we analyze the distribution of the positive sojourn time
conditioned on return to the origin
from the viewpoint of generating functions.
The purpose of this section is to understand the phenomenon that, at time $n=4m$,
\[
\mu_{4m}(2),\mu_{4m}(4),\ldots,\mu_{4m}(4m-2)
\]
are exactly uniform from the algebraic properties of the coin matrix.
In particular, in order to show that this uniform distribution
is specific to the Hadamard coin,
we first study the structure of the generating function
with matrix entries for paths returning to the origin.

Here it is important to distinguish
the matrix-valued generating function $\Gamma(z,t)$
from the generating function of probability weights $w_n(k)$.
We define the generating function with matrix entries by
\[
\Gamma(z,t)
:=
\sum_{n\ge0}\sum_{k=0}^{n}\Gamma_n(k)z^nt^k.
\]
On the other hand, the generating function of probability weights is
\[
W(z,t)
:=
\sum_{n\ge0}\sum_{k=0}^{n}w_n(k)z^nt^k.
\]
It should be noted that $W(z,t)$ is not simply a matrix entry of $\Gamma(z,t)$.
Indeed,
\[
w_n(k)=\|\Gamma_n(k)\varphi_\ast\|^2
=
\varphi_\ast^\ast \Gamma_n(k)^\ast\Gamma_n(k)\varphi_\ast,
\]
and it is defined by taking the squared norm of each coefficient
$\Gamma_n(k)$.
Therefore, in this section, we first clarify the algebraic structure
of $\Gamma(z,t)$.
Then we consider the conditional distribution $\mu_n(k)$
through the squared norm of each coefficient.

In this paper, we first treat the essential case of the rotation coin
\[
U(\theta)=
\begin{bmatrix}
\cos\theta & \sin\theta\\
\sin\theta & -\cos\theta
\end{bmatrix}
\qquad
\left(0<\theta<\frac{\pi}{2}\right).
\]
We decompose this coin as
\[
U(\theta)=P+Q,
\]
where
\[
P=
\begin{bmatrix}
\cos\theta & \sin\theta\\
0 & 0
\end{bmatrix},
\qquad
Q=
\begin{bmatrix}
0 & 0\\
\sin\theta & -\cos\theta
\end{bmatrix}.
\]
Here $P$ corresponds to a move to the left,
and $Q$ corresponds to a move to the right.

We first record a simple invariance which will be used later.

\begin{lemma}
\label{lem:global-phase}
If the coin matrix $U$ is replaced by
\[
U'=e^{i\eta}U
\qquad (\eta\in\mathbb R),
\]
then the weights $w_n(k)$ before normalization
and the conditional distribution $\mu_n(k)$ do not change.
\end{lemma}

\begin{proof}
If the coin matrix is replaced by $U'=e^{i\eta}U$,
then the corresponding matrices become
\[
P'=e^{i\eta}P,
\qquad
Q'=e^{i\eta}Q.
\]
Let $A(w)$ be the matrix product corresponding to a path $w$ of length $n$.
Then the transformed matrix product is
\[
A'(w)=e^{in\eta}A(w).
\]
Therefore, the sum over paths which return to the origin
and have positive sojourn time $k$ becomes
\[
\Gamma_n'(k)=e^{in\eta}\Gamma_n(k).
\]
Hence
\[
\|\Gamma_n'(k)\varphi_\ast\|^2
=
\|e^{in\eta}\Gamma_n(k)\varphi_\ast\|^2
=
\|\Gamma_n(k)\varphi_\ast\|^2.
\]
Thus $w_n(k)$ does not change.
Consequently, its normalization $\mu_n(k)$ also does not change.
\end{proof}

\begin{lemma}
\label{lem:left-diagonal}
If the coin matrix $U$ is replaced by
\[
U'=D_1U,
\qquad
D_1=
\begin{bmatrix}
e^{i\alpha} & 0\\
0 & e^{i\beta}
\end{bmatrix},
\]
then the weights $w_n(k)$ before normalization for paths returning to the origin
and the conditional distribution $\mu_n(k)$ do not change.
\end{lemma}

\begin{proof}
Under this transformation, we have
\[
P'=e^{i\alpha}P,
\qquad
Q'=e^{i\beta}Q.
\]
For a path of length $n$ which returns to the origin,
the numbers of moves to the left and to the right are equal,
and each of them is $n/2$.
Therefore, for any path $w$ returning to the origin,
the corresponding matrix product is transformed as
\[
A'(w)
=
e^{i\alpha n/2}e^{i\beta n/2}A(w)
=
e^{i(\alpha+\beta)n/2}A(w).
\]
This phase factor is common to all paths returning to the origin
at the same time $n$.
Therefore,
\[
\Gamma_n'(k)
=
e^{i(\alpha+\beta)n/2}\Gamma_n(k),
\]
and hence
\[
\|\Gamma_n'(k)\varphi_\ast\|^2
=
\|\Gamma_n(k)\varphi_\ast\|^2.
\]
Thus $w_n(k)$ and $\mu_n(k)$ do not change.
\end{proof}

The above lemmas show that the global phase
and a diagonal unitary transformation from the left
do not have an essential effect on the positive sojourn time distribution
conditioned on return to the origin.
In this paper, a coin which is transformed into the Hadamard coin
by a global phase and a diagonal unitary transformation from the left
is called a coin of Hadamard type in this restricted sense.
However, in the main theorem below,
we first characterize the condition for exact uniformity
in the family of rotation coins $U(\theta)$.

The next lemma gives the structure of the generating function
which plays the central role in this section.
The important point is that the generating function with matrix entries
for paths returning to the origin is described,
through an absorbing process on the half line,
by the characteristic roots of a quadratic equation.

\begin{lemma}
\label{lem:matrix-gf}
Under the rotation coin $U(\theta)$,
the matrix-valued generating function $\Gamma(z,t)$
for paths returning to the origin is algebraic in $z$.
More concretely, its entries can be expressed by using the two square roots
\[
\sqrt{
z^4t^4+2\cos(2\theta)z^2t^2+1
}
\]
and
\[
\sqrt{
z^4+2\cos(2\theta)z^2+1
}.
\]
\end{lemma}

\begin{proof}
We decompose the rotation coin as
\[
U(\theta)=P+Q,
\]
where
\[
P=
\begin{bmatrix}
\cos\theta & \sin\theta\\
0&0
\end{bmatrix},
\qquad
Q=
\begin{bmatrix}
0&0\\
\sin\theta&-\cos\theta
\end{bmatrix}.
\]
Here $P$ corresponds to a move to the left,
and $Q$ corresponds to a move to the right.

A nonempty path returning to the origin can be decomposed uniquely
as a concatenation of first-return excursions.
Here a first-return excursion means a path which starts from the origin,
returns to the origin at its terminal time,
and does not visit the origin at intermediate times.
Let $E_\theta(z,t)$ be the matrix-valued generating function
for one first-return excursion.
Then, including the empty path, we have
\[
\Gamma(z,t)
=
I+E_\theta(z,t)+E_\theta(z,t)^2+\cdots
=
(I-E_\theta(z,t))^{-1}.
\]
Here the inverse is understood as an inverse of formal power series,
since $E_\theta(z,t)$ has no constant term.

Therefore, the algebraic structure of $\Gamma(z,t)$
is determined by the algebraic structure of $E_\theta(z,t)$.

We first consider a positive side excursion.
This is a process which enters from the origin to $+1$,
then moves in the half line $\{1,2,\ldots\}$,
and finally returns from $1$ to $0$.
For $x\ge1$, define the generating function $F_x(z,t)$
with matrix entries as the sum over paths which start from position $x$,
move in the half line, and are finally absorbed at $0$.

According to the convention of this paper that matrix products act
from right to left,
and since each step in the positive side has weight $t$,
for $x\ge2$ we have
\[
F_x(z,t)
=
zt\left(F_{x-1}(z,t)P+F_{x+1}(z,t)Q\right).
\]
At the boundary $x=1$, we have
\[
F_1(z,t)
=
ztF_2(z,t)Q+ztP.
\]
Here $ztP$ is the contribution of the last one step
which is absorbed from position $1$ to position $0$.

We put a characteristic solution in the form
\[
F_x(z,t)=r^xV,
\]
where $V$ is a nonzero $2\times2$ matrix.
Substituting this form into the recurrence for $x\ge2$, we obtain
\[
r^xV
=
zt\left(r^{x-1}VP+r^{x+1}VQ\right).
\]
Dividing by $r^{x-1}$, we get
\[
rV
=
zt(VP+r^2VQ).
\]
Therefore
\[
V(rI-ztP-ztr^2Q)=0.
\]
For a nontrivial solution to exist, we must have
\[
\det(rI-ztP-ztr^2Q)=0.
\]

Substituting the matrices, we have
\[
rI-ztP-ztr^2Q
=
\begin{bmatrix}
r-zt\cos\theta & -zt\sin\theta\\
-ztr^2\sin\theta & r+ztr^2\cos\theta
\end{bmatrix}.
\]
Hence
\[
\begin{aligned}
\det(rI-ztP-ztr^2Q)
&=
(r-zt\cos\theta)(r+ztr^2\cos\theta)
-(zt)^2r^2\sin^2\theta\\
&=
r^2-(zt)r\cos\theta+(zt)r^3\cos\theta
-(zt)^2r^2(\cos^2\theta+\sin^2\theta)\\
&=
r^2-(zt)r\cos\theta+(zt)r^3\cos\theta-(zt)^2r^2.
\end{aligned}
\]
Assuming $r\neq0$ and dividing by $r$, we obtain
\[
r-(zt)\cos\theta+(zt)r^2\cos\theta -(zt)^2r=0.
\]
That is,
\[
(zt\cos\theta)r^2+\left(1-(zt)^2\right)r-(zt\cos\theta)=0.
\]
The discriminant of this quadratic equation is
\[
\begin{aligned}
\left(1-(zt)^2\right)^2
+4(zt)^2\cos^2\theta
&=
1-2z^2t^2+z^4t^4+4z^2t^2\cos^2\theta\\
&=
1+2(2\cos^2\theta-1)z^2t^2+z^4t^4\\
&=
1+2\cos(2\theta)z^2t^2+z^4t^4.
\end{aligned}
\]
Thus, if we put
\[
\Delta(z,t;\theta)
=
z^4t^4+2\cos(2\theta)z^2t^2+1,
\]
then the characteristic roots are given by
\[
r_{\pm}(z,t;\theta)
=
\frac{(zt)^2-1\pm\sqrt{\Delta(z,t;\theta)}}{2zt\cos\theta}.
\]

Let $r_0(z,t;\theta)$ be the branch which is a formal power series
around $z=0$.
Namely, we put
\[
r_0(z,t;\theta)
=
\frac{(zt)^2-1+\sqrt{\Delta(z,t;\theta)}}{2zt\cos\theta}.
\]
Indeed, expanding the square root around $z=0$, we have
\[
\sqrt{\Delta(z,t;\theta)}
=
1+\cos(2\theta)z^2t^2+O(z^4).
\]
Hence
\[
\begin{aligned}
r_0(z,t;\theta)
&=
\frac{z^2t^2-1+1+\cos(2\theta)z^2t^2+O(z^4)}
{2zt\cos\theta}\\
&=
\frac{(1+\cos(2\theta))z^2t^2+O(z^4)}
{2zt\cos\theta}\\
&=
zt\cos\theta+O(z^3).
\end{aligned}
\]
Thus, $r_0$ can be treated as a formal power series
which is regular around $z=0$.

Next, we study the algebraic structure of $F_1(z,t)$
from the boundary condition.
Substituting $F_x=r_0^xV$ into
\[
F_1=ztF_2Q+ztP,
\]
we obtain
\[
r_0V=zt\,r_0^2VQ+ztP.
\]
Therefore,
\[
V(r_0I-ztr_0^2Q)=ztP.
\]
Here
\[
r_0I-ztr_0^2Q
=
\begin{bmatrix}
r_0 & 0\\
-ztr_0^2\sin\theta & r_0+ztr_0^2\cos\theta
\end{bmatrix}.
\]
Thus, the inverse of this matrix is expressed as a rational function
of $r_0$, $z$, $t$, $\sin\theta$, and $\cos\theta$.
Therefore,
\[
V=ztP(r_0I-ztr_0^2Q)^{-1},
\]
and each entry of $V$ is expressed as a rational function of $r_0$.
It follows that each entry of
\[
F_1(z,t)=r_0V
\]
is also expressed as a rational function of $r_0$.

From the above discussion, the absorbing part of a positive side excursion
is an algebraic function expressed by using $r_0$, and hence by using
\[
\sqrt{\Delta(z,t;\theta)}.
\]
Moreover, the whole contribution of a positive side excursion is obtained
by multiplying the contribution of the first one step
which enters from the origin to $+1$.
Under the convention for matrix products in this paper, this contribution is
\[
E_\theta^{(+)}(z,t)=F_1(z,t)\,ztQ.
\]
Here $ztQ$ is the product of $z,t$ and the matrix $Q$.
Therefore, the algebraic structure involving the square root
is the same as that of $F_1(z,t)$.

The negative side excursion is treated in the same way.
On the negative side, the weight $t$ for the positive sojourn time
does not appear.
Hence it is reduced to the same problem on the half line
with $t=1$.
Therefore, the expression for the negative side excursion contains
\[
\sqrt{z^4+2\cos(2\theta)z^2+1}.
\]
More concretely, if the whole contribution of a negative side excursion
is denoted by $E_\theta^{(-)}(z,1)$,
then each entry is obtained from the entries of the above absorbing process
on the half line by putting $t=1$.

Consequently, each entry of the generating function for one excursion
\[
E_\theta(z,t)
=
E_\theta^{(+)}(z,t)+E_\theta^{(-)}(z,1)
\]
is an algebraic function expressed by using
\[
\sqrt{z^4t^4+2\cos(2\theta)z^2t^2+1}
\]
and
\[
\sqrt{z^4+2\cos(2\theta)z^2+1}.
\]
Finally, since
\[
\Gamma(z,t)=(I-E_\theta(z,t))^{-1},
\]
each entry of $\Gamma(z,t)$ is also algebraic in $z$
and is expressed by using the same square roots.
\end{proof}

By Lemma \ref{lem:matrix-gf},
the algebraic structure of the generating function for the rotation coin
$U(\theta)$ is governed by
\[
\Delta(z,t;\theta)
=
z^4t^4+2\cos(2\theta)z^2t^2+1.
\]
The Hadamard coin corresponds to
\[
\theta=\frac{\pi}{4}.
\]
In this case,
\[
\cos(2\theta)=0.
\]
Therefore,
\[
\Delta\left(z,t;\frac{\pi}{4}\right)
=
1+z^4t^4.
\]
This simplification shows that the coefficient structure at time $4m$
becomes especially simple in the case of the Hadamard coin.
However, in order to prove the necessity of the exact uniform distribution,
one needs not only the form of the square root
but also an actual comparison of coefficients.

We next confirm that this simplification of the square root occurs
only in the Hadamard case.

\begin{lemma}
\label{lem:sqrt-reduction}
For a real number $\theta$,
the power series expansion around $z=0$ of
\[
\sqrt{z^4+2\cos(2\theta)z^2+1}
\]
consists only of powers of the form $z^{4m}$
if and only if
\[
\cos(2\theta)=0.
\]
\end{lemma}

\begin{proof}
Put
\[
c=2\cos(2\theta).
\]
Then
\[
\sqrt{z^4+2\cos(2\theta)z^2+1}
=
\sqrt{1+cz^2+z^4}.
\]
Expanding the square root around $z=0$, we obtain
\[
\sqrt{1+cz^2+z^4}
=
1+\frac{c}{2}z^2+O(z^4).
\]
Therefore, if $c\ne0$, then the term of degree $2$ in $z$ appears.
In this case, the power series expansion does not consist only of powers
of the form $z^{4m}$.

On the other hand, if $c=0$, that is, if $\cos(2\theta)=0$, then
\[
\sqrt{1+cz^2+z^4}
=
\sqrt{1+z^4}.
\]
Its power series expansion is
\[
\sqrt{1+z^4}
=
1+\frac12z^4-\frac18z^8+\cdots.
\]
Hence only powers of the form $z^{4m}$ appear.
This proves the assertion.
\end{proof}

This lemma shows that only in the Hadamard case
the square root is simplified to
\[
\sqrt{1+z^4}.
\]
However, in order to prove the necessity of the exact uniform distribution,
one needs not only the form of the square root
but also an actual comparison of coefficients.
Therefore, below we show that the Hadamard type is necessary
by comparing coefficients of low degree at time $8$.

\subsection{Four point extraction and the characterization theorem}
\label{subsec:four-point}

We first introduce a basic operation for extracting the terms
whose times are multiples of $4$.

\begin{lemma}
\label{lem35}
Let
\[
W(z,t):=\sum_{n\ge0}\sum_{k=0}^{n}w_n(k)z^nt^k
\]
be the generating function in two variables for the weights before normalization.
Then the four point average
\[
H_W(z,t):=
\frac14\Bigl(W(z,t)+W(iz,t)+W(-z,t)+W(-iz,t)\Bigr)
\]
keeps only the terms with $n\equiv0\pmod4$
and eliminates the terms of all other degrees.
Namely,
\[
H_W(z,t)
=
\sum_{m\ge0}\sum_{k=0}^{4m}w_{4m}(k)z^{4m}t^k.
\]
\end{lemma}

\begin{proof}
We compute coefficient by coefficient as formal power series.
Put $\omega=i$.
Under the substitution $z\mapsto \omega z$, the monomial $z^n$ is sent to
\[
(\omega z)^n=\omega^n z^n=i^nz^n.
\]
Therefore,
\[
\begin{aligned}
H_W(z,t)
&=
\frac14\Bigl(W(z,t)+W(iz,t)+W(-z,t)+W(-iz,t)\Bigr)\\
&=
\frac14
\sum_{n\ge0}\sum_{k=0}^{n}w_n(k)
\Bigl\{z^n+(iz)^n+(-z)^n+(-iz)^n\Bigr\}t^k\\
&=
\frac14
\sum_{n\ge0}\sum_{k=0}^{n}w_n(k)
\bigl(1+i^n+(-1)^n+(-i)^n\bigr)z^nt^k.
\end{aligned}
\]
Here
\[
1+i^n+(-1)^n+(-i)^n
=
\begin{cases}
4 & (n\equiv0\pmod4),\\
0 & (n\not\equiv0\pmod4).
\end{cases}
\]
Thus, only the terms with $n\equiv0\pmod4$ remain in $H_W(z,t)$.
\end{proof}

By Lemma \ref{lem35},
$H_W(z,t)$ is the generating function which extracts
only the information on times which are multiples of $4$
from the weights before normalization.
For each $m\ge0$, put
\[
G_m(t):=\sum_{k=0}^{4m}w_{4m}(k)t^k.
\]
Then
\[
H_W(z,t)=\sum_{m\ge0}G_m(t)z^{4m}.
\]
Moreover, the distribution of the positive sojourn time
conditioned on return to the origin is given by
\[
\mu_{4m}(k)
=
\frac{w_{4m}(k)}{\sum_{j=0}^{4m}w_{4m}(j)}
=
\frac{[t^k]G_m(t)}{G_m(1)}.
\]

Therefore, the exact uniform distribution at time $4m$
can be written in terms of the generating polynomial $G_m(t)$ as
\[
\frac{G_m(t)}{G_m(1)}
=
\frac{t^2+t^4+\cdots+t^{4m-2}}{2m-1}.
\]
This expression is important in order to distinguish
the weights before normalization from the conditional distribution.

In the case of the Hadamard coin,
Konno \cite{Konno2012} proved that, at time $4m$,
\[
\mu_{4m}(2j)=\frac{1}{2m-1}
\qquad
(j=1,2,\ldots,2m-1)
\]
holds, and that
\[
\mu_{4m}(k)=0
\]
for all other $k$.
Below, we show that, in the family of rotation coins,
this exact uniform distribution occurs only for the Hadamard coin.

\begin{lemma}
\label{lem37}
Let the initial state be
\[
\varphi_\ast=
\frac{1}{\sqrt2}
\begin{bmatrix}
1\\ i
\end{bmatrix}.
\]
For the rotation coin
\[
U(\theta)=
\begin{bmatrix}
\cos\theta & \sin\theta\\
\sin\theta & -\cos\theta
\end{bmatrix}
\qquad
\left(0<\theta<\frac{\pi}{2}\right),
\]
if
\[
\mu_8(2)=\mu_8(4)=\mu_8(6)
\]
holds at time $8$, then
\[
\cos(2\theta)=0.
\]
Consequently,
\[
\theta=\frac{\pi}{4}.
\]
\end{lemma}

\begin{proof}
Put $c=\cos\theta$ and $s=\sin\theta$.
Since $0<\theta<\pi/2$, we have $c>0$ and $s>0$.

We compute the weights before normalization at time $8$.
This computation is obtained by classifying all paths returning to the origin
of length $8$ according to the positive sojourn time,
summing all corresponding matrix products,
applying the resulting matrices to the initial state $\varphi_\ast$,
and then taking the squared norm.
The paths returning to the origin at time $8$
consist of four moves to the left and four moves to the right.
Hence there are
\[
\binom{8}{4}=70
\]
such paths in total.

This finite computation can be reproduced by the following recurrence.
Let $\mathcal A_j(x,k)$ be the matrix weight corresponding to time $j$,
position $x$, and positive sojourn time $k$.
The initial condition is
\[
\mathcal A_0(0,0)=I,
\]
and all other $\mathcal A_0(x,k)$ are the zero matrix.
For $\varepsilon\in\{-1,+1\}$, put
\[
\chi(x,\varepsilon)
=
\begin{cases}
1, & x>0,\\
1, & x=0 \text{ and } \varepsilon=+1,\\
0, & \text{otherwise}.
\end{cases}
\]
This function represents whether the interval $[j,j+1)$
is assigned to the positive side.
Since $P$ corresponds to a move to the left
and $Q$ corresponds to a move to the right,
the matrix weight at time $j+1$ is obtained
by summing all contributions which arrive from time $j$.
Namely, for each $y,\ell$, we update by
\[
\mathcal A_{j+1}(y,\ell)
=
\sum_{\substack{x,k\\ y=x-1\\ \ell=k+\chi(x,-1)}}
P\,\mathcal A_j(x,k)
+
\sum_{\substack{x,k\\ y=x+1\\ \ell=k+\chi(x,+1)}}
Q\,\mathcal A_j(x,k).
\]
Then
\[
\Gamma_n(k)=\mathcal A_n(0,k).
\]
The formulas below at time $8$ are obtained by applying this recurrence
for $j=0,\ldots,7$
and simplifying the result by using $c^2+s^2=1$.

As a result, we obtain
\[
\begin{aligned}
w_8(0)=w_8(8)
&=
\frac12
s^2(c^2-s^2)^2
\left(c^4-5c^2s^2+s^4\right)^2,\\
w_8(2)=w_8(6)
&=
c^2s^4
\left(c^8-4c^6s^2+7c^4s^4-4c^2s^6+s^8\right),\\
w_8(4)
&=
c^2s^4
\left(c^8-5c^6s^2+9c^4s^4-5c^2s^6+s^8\right).
\end{aligned}
\]

The return probability to the origin
\[
\sum_{j=0}^{8}w_8(j)
\]
is positive, and $\mu_8(k)$ is obtained by dividing $w_8(k)$
by this common normalization constant.
Therefore, by the assumption, in particular we have
\[
w_8(2)=w_8(4).
\]
From the above expressions,
\[
\begin{aligned}
w_8(2)-w_8(4)
&=
c^2s^4
\Bigl[
c^8-4c^6s^2+7c^4s^4-4c^2s^6+s^8 \\
&\qquad
-\left(c^8-5c^6s^2+9c^4s^4-5c^2s^6+s^8\right)
\Bigr]\\
&=
c^2s^4\left(c^6s^2-2c^4s^4+c^2s^6\right)\\
&=
c^4s^6(c^2-s^2)^2.
\end{aligned}
\]
Thus,
\[
c^4s^6(c^2-s^2)^2=0.
\]
Since $c>0$ and $s>0$, it follows that
\[
c^2=s^2.
\]
That is,
\[
\cos^2\theta=\sin^2\theta.
\]
Hence
\[
\theta=\frac{\pi}{4}.
\]
This is equivalent to
\[
\cos(2\theta)=0.
\]
\end{proof}

\begin{definition}
\label{def:exact-uniform}
Let $m\ge1$.
We say that the distribution of the positive sojourn time
conditioned on return to the origin is exactly uniform at time $4m$
if
\[
\mu_{4m}(2)=\mu_{4m}(4)=\cdots=\mu_{4m}(4m-2)
=
\frac{1}{2m-1}
\]
holds, and if
\[
\mu_{4m}(k)=0
\]
for all other $k$.
\end{definition}

\begin{theorem}
\label{thm31}
Let the initial state be
\[
\varphi_\ast=
\frac{1}{\sqrt2}
\begin{bmatrix}
1\\ i
\end{bmatrix}.
\]
For the rotation coin
\[
U(\theta)=
\begin{bmatrix}
\cos\theta & \sin\theta\\
\sin\theta & -\cos\theta
\end{bmatrix}
\qquad
\left(0<\theta<\frac{\pi}{2}\right),
\]
the following three conditions are equivalent.
\begin{enumerate}
\item[\rm (i)]
The distribution of the positive sojourn time
conditioned on return to the origin is exactly uniform at time $8$.

\item[\rm (ii)]
For every $m\ge2$, the distribution of the positive sojourn time
conditioned on return to the origin is exactly uniform at time $4m$.

\item[\rm (iii)]
\[
\theta=\frac{\pi}{4}.
\]
\end{enumerate}
\end{theorem}

\begin{proof}
First, assume (iii).
Then
\[
U(\theta)
=
\frac1{\sqrt2}
\begin{bmatrix}
1&1\\
1&-1
\end{bmatrix},
\]
and this is the Hadamard coin.
Therefore, by Konno's result \cite{Konno2012}
for the Hadamard walk with the initial state considered here,
for every $m\ge1$,
\[
\mu_{4m}(2j)=\frac{1}{2m-1}
\qquad
(j=1,2,\ldots,2m-1)
\]
holds, and
\[
\mu_{4m}(k)=0
\]
holds for all other $k$.
Thus, in particular, (i) and (ii) hold.

Next, assume (ii).
Then, in particular, the exact uniform distribution holds
for $m=2$, that is, at time $8$.
Therefore, (i) holds.

Finally, assume (i).
If the exact uniform distribution holds at time $8$,
then in particular
\[
\mu_8(2)=\mu_8(4)=\mu_8(6).
\]
Hence, by Lemma \ref{lem37}, we have
\[
\cos(2\theta)=0.
\]
Since $0<\theta<\pi/2$, it follows that $\theta=\pi/4$.
Thus, (iii) holds.

Therefore, (i), (ii), and (iii) are equivalent.
\end{proof}

\begin{corollary}
\label{cor:hadamard-type}
Every coin obtained from the Hadamard coin
by multiplication by a global phase and by left multiplication
by a diagonal unitary matrix gives the same return-conditioned
positive sojourn-time distribution as the Hadamard coin.

Conversely, within the rotation-coin family $U(\theta)$,
exact uniformity at every time $4m$ with $m\ge2$
occurs only for $\theta=\pi/4$.
Thus, modulo these distribution-preserving transformations,
the finite-time uniformity phenomenon is specific to coins of Hadamard type.
\end{corollary}

\begin{proof}
The first assertion follows from Lemma \ref{lem:global-phase}
and Lemma \ref{lem:left-diagonal}.
The converse assertion within the rotation-coin family follows from
Theorem \ref{thm31}.
\end{proof}

The results of this section show that the uniform distribution at finite times
for the Hadamard coin is not a mere coincidence
or a numerical phenomenon.
It is characterized by the rigidity condition
\[
\cos(2\theta)=0
\]
on the parameter of the rotation coin.
In other words, the exact uniform distribution at finite time $n=4m$
is an exceptional phenomenon which occurs only for the Hadamard coin
within the family of rotation coins.

\section{Numerical examples}
\label{sec:numerical}

In this section, we illustrate Theorem \ref{thm31}
by the example at time $8$.
The numerical examples and figures in this section
are not used in the proof of the theorem.
They are included in order to visually confirm that
the exact uniform distribution appears only in the case of the Hadamard coin.

We take the initial state to be
\[
\varphi_\ast=
\frac{1}{\sqrt2}
\begin{bmatrix}
1\\ i
\end{bmatrix}.
\]
We also consider the rotation coin
\[
U(\theta)=
\begin{bmatrix}
\cos\theta & \sin\theta\\
\sin\theta & -\cos\theta
\end{bmatrix}.
\]

We first consider the case $\theta=\pi/4$,
which corresponds to the Hadamard coin.
In this case,
\[
P=
\frac{1}{\sqrt2}
\begin{bmatrix}
1&1\\
0&0
\end{bmatrix},
\qquad
Q=
\frac{1}{\sqrt2}
\begin{bmatrix}
0&0\\
1&-1
\end{bmatrix}.
\]
At time $8$, a path returning to the origin consists of
four moves to the left and four moves to the right.
Hence the total number of such paths is
\[
\binom{8}{4}=70.
\]
Classifying these paths according to the positive sojourn time $k$,
and computing the matrix weights $\Gamma_8(k)$
according to the definition in Section \ref{sec:definition},
we obtain
\[
\Gamma_8(0)=O,\qquad
\Gamma_8(8)=O,
\]
and
\[
\Gamma_8(2)
=
\Gamma_8(4)
=
\Gamma_8(6)
=
\frac1{16}
\begin{bmatrix}
1&1\\
-1&1
\end{bmatrix}.
\]
Therefore, for example,
\[
\begin{aligned}
\Gamma_8(2)\varphi_\ast
&=
\frac1{16}
\begin{bmatrix}
1&1\\
-1&1
\end{bmatrix}
\frac1{\sqrt2}
\begin{bmatrix}
1\\ i
\end{bmatrix}
=
\frac1{16\sqrt2}
\begin{bmatrix}
1+i\\
-1+i
\end{bmatrix}.
\end{aligned}
\]
Hence
\[
\begin{aligned}
w_8(2)
&=
\|\Gamma_8(2)\varphi_\ast\|^2\\
&=
\left\|
\frac1{16\sqrt2}
\begin{bmatrix}
1+i\\
-1+i
\end{bmatrix}
\right\|^2\\
&=
\frac1{512}
\left(|1+i|^2+|-1+i|^2\right)
=
\frac1{128}.
\end{aligned}
\]
Similarly,
\[
w_8(4)=w_8(6)=\frac1{128},
\qquad
w_8(0)=w_8(8)=0.
\]
Thus,
\[
\sum_{k=0}^{8}w_8(k)=\frac3{128}.
\]
After normalization, we obtain
\[
\mu_8(2)=\mu_8(4)=\mu_8(6)=\frac13,
\qquad
\mu_8(0)=\mu_8(8)=0.
\]
Therefore, in the case of the Hadamard coin,
the distribution of the positive sojourn time
conditioned on return to the origin
is exactly uniform on $\{2,4,6\}$.

This agrees with the value
\[
\frac{1}{2m-1}=\frac13
\]
in Definition \ref{def:exact-uniform} when $m=2$.
Thus, this computation concretely confirms
the assertion of Theorem \ref{thm31} at time $8$.

We next look at examples of rotation coins
which are not the Hadamard coin.
Using the same definition, we classified the paths
which return to the origin at time $8$,
computed the matrix weights for each $k$,
and then normalized them.

First, when $\theta=\pi/5$,
the conditional distribution at time $8$ is
\[
\begin{array}{c|ccccc}
k & 0 & 2 & 4 & 6 & 8\\
\hline
\mu_8(k)
& 0.237043 & 0.199088 & 0.127738 & 0.199088 & 0.237043
\end{array}
\]
Thus,
\[
\mu_8(2),\quad \mu_8(4),\quad \mu_8(6)
\]
are not equal, and the exact uniform distribution does not hold.

Moreover, when $\theta=\pi/6$, we have
\[
\begin{array}{c|ccccc}
k & 0 & 2 & 4 & 6 & 8\\
\hline
\mu_8(k)
& 0.173010 & 0.259516 & 0.134948 & 0.259516 & 0.173010
\end{array}
\]
In this case also,
\[
\mu_8(2),\quad \mu_8(4),\quad \mu_8(6)
\]
are not equal, and the exact uniform distribution does not hold.

The above numerical examples are obtained by directly enumerating
the paths which return to the origin at time $8$,
according to the definition in Section \ref{sec:definition}.
In the case of the Hadamard coin,
a completely flat distribution appears on $\{2,4,6\}$.
On the other hand, for rotation coins which are not the Hadamard coin,
some mass remains also at $k=0$ and $k=8$,
and moreover
\[
\mu_8(2),\quad \mu_8(4),\quad \mu_8(6)
\]
are not equal.
These direct computations agree with the assertion of
Theorem \ref{thm31},
which states that the exact uniform distribution at time $8$
is characterized by the rigidity condition
$\theta=\pi/4$.

These behaviors are illustrated in Figure \ref{fig:hadamard-time8}.
\begin{figure}[H]
\centering
\begin{minipage}{0.32\linewidth}
	\centering
	\includegraphics[width=\linewidth]{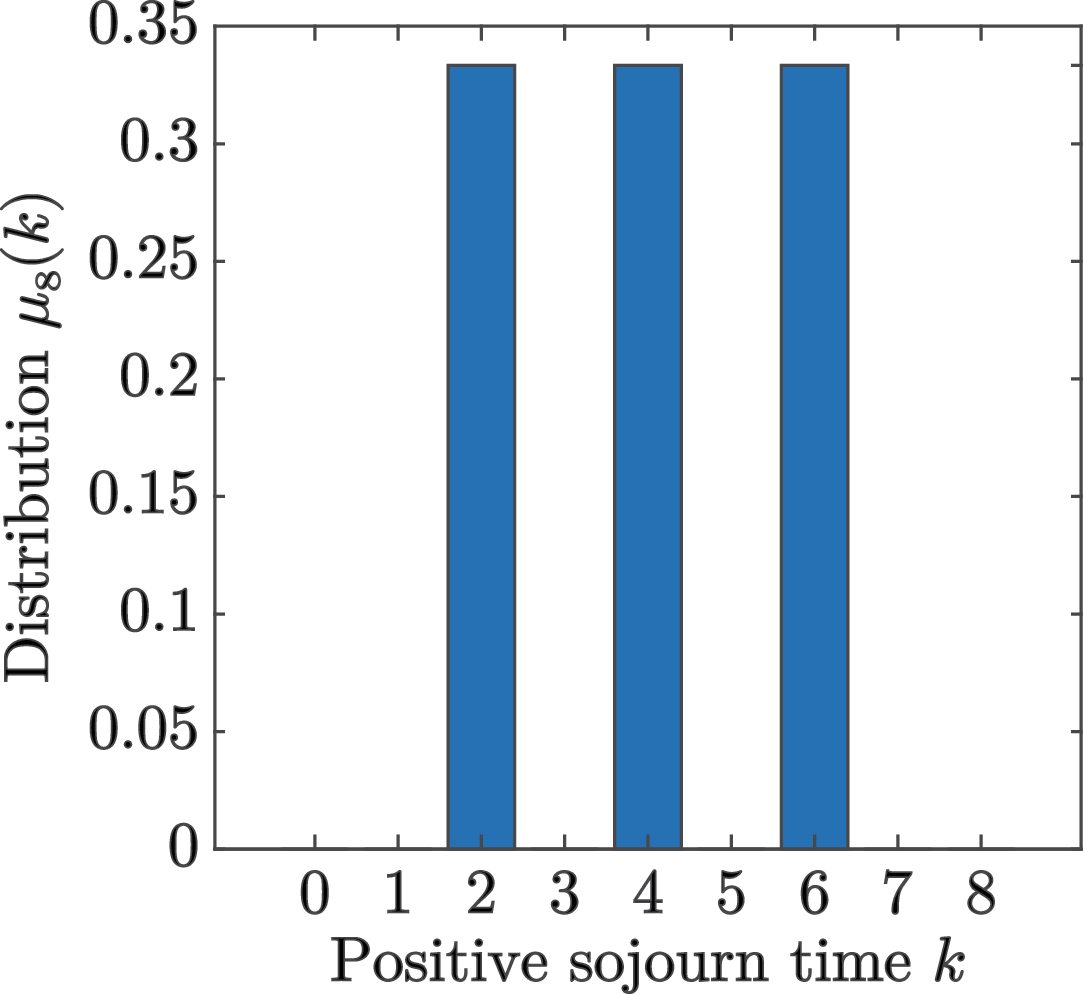}

	{\small (a) $\theta=\pi/4$}
\end{minipage}\hfill
\begin{minipage}{0.32\linewidth}
	\centering
	\includegraphics[width=\linewidth]{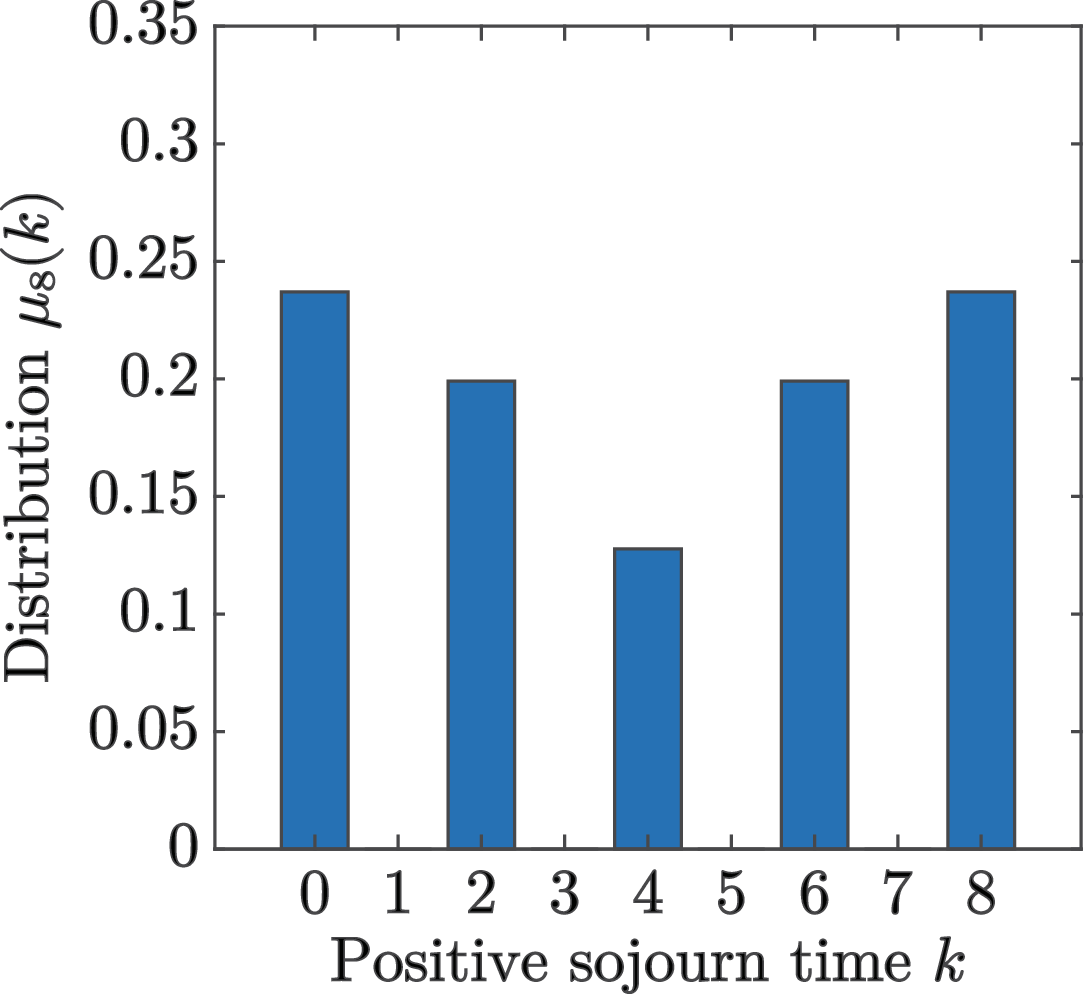}

	{\small (b) $\theta=\pi/5$}
\end{minipage}\hfill
\begin{minipage}{0.32\linewidth}
	\centering
	\includegraphics[width=\linewidth]{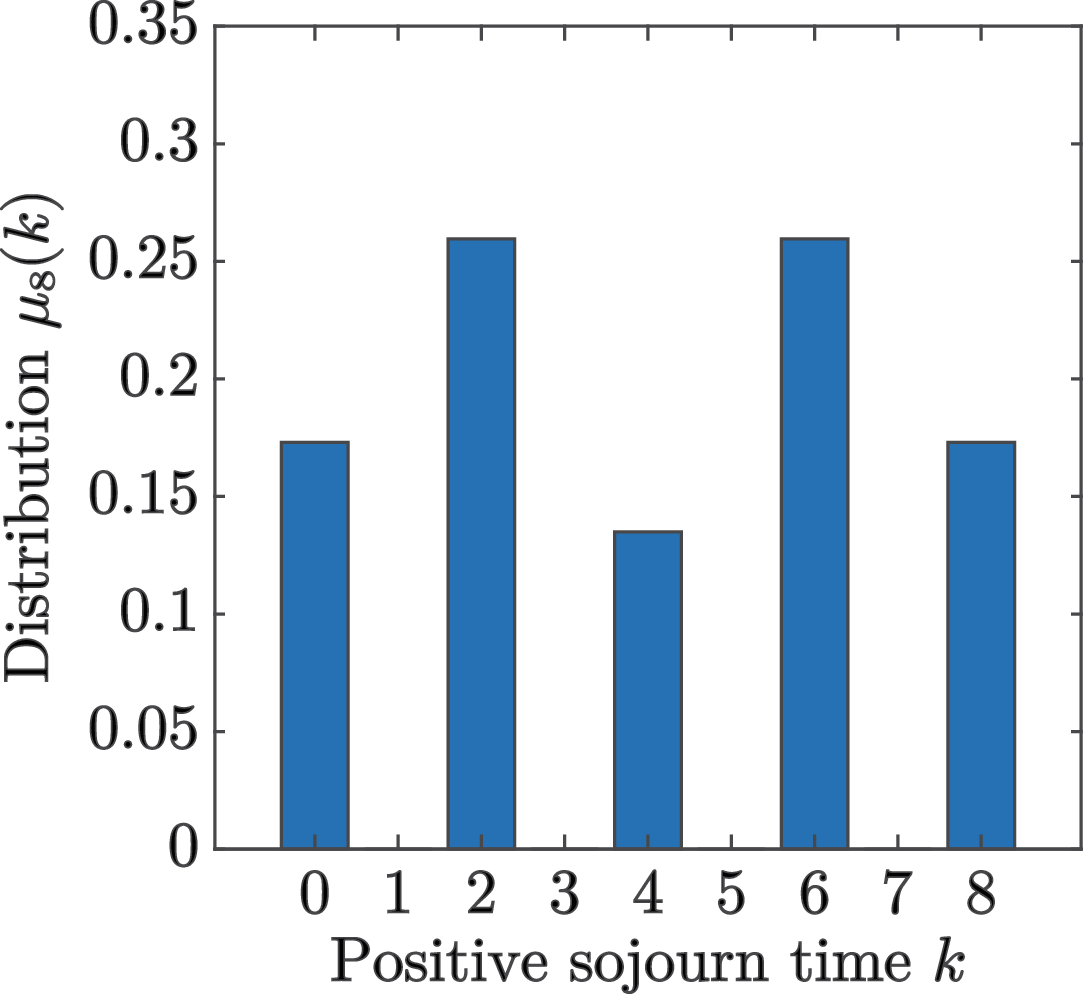}

	{\small (c) $\theta=\pi/6$}
\end{minipage}
\caption{The distribution of the positive sojourn time
conditioned on return to the origin at time $8$
for the rotation coin $U(\theta)$
with $\theta=\pi/4,\pi/5,\pi/6$.
When $\theta=\pi/4$, the distribution is exactly uniform
on $k=2,4,6$, and it is $0$ at $k=0,8$.
On the other hand, when $\theta=\pi/5$ or $\theta=\pi/6$,
this uniformity breaks down.}
\label{fig:hadamard-time8}
\end{figure}

\section{Conclusion and future problems}

In this paper, we studied the distribution of the positive sojourn time
conditioned on return to the origin
for a quantum walk on the line with two internal states.
We investigated how the exact uniform distribution at finite times,
which was shown by Konno \cite{Konno2012} in the case of the Hadamard coin,
is characterized in the family of rotation coins.

As the main result, for the rotation coin
\[
U(\theta)=
\begin{bmatrix}
\cos\theta & \sin\theta\\
\sin\theta & -\cos\theta
\end{bmatrix}
\qquad
\left(0<\theta<\frac{\pi}{2}\right)
\]
and the initial state
\[
\varphi_\ast=
\frac{1}{\sqrt2}
\begin{bmatrix}
1\\ i
\end{bmatrix},
\]
we proved that the following three conditions are equivalent:
the exact uniformity at time $8$,
the exact uniformity at every time which is a multiple of $4$
and is at least $8$,
and $\theta=\pi/4$.
Therefore, in the family of rotation coins,
the exact uniformity of the positive sojourn time distribution
conditioned on return to the origin
occurs only for the Hadamard coin.
In particular, it is not necessary to check all times;
the uniformity at time $8$ already forces the coin
to be the Hadamard coin.

In the proof, we introduced a generating function with matrix entries
for paths returning to the origin,
and studied its algebraic structure.
In the case of the Hadamard coin,
the square root appearing in the generating function
is simplified in a special way.
However, the necessity was proved by comparing coefficients
of low degree at time $8$.
This shows that the exact uniform distribution at finite times
is not a mere numerical phenomenon,
but a rigidity phenomenon coming from the structure of the coin matrix.

As a future problem, it is natural to extend the result
to general unitary coins with two internal states.
In this paper, we treated the family of rotation coins.
For a general coin given by a $2\times 2$ unitary matrix, it would be natural to investigate
whether coins of Hadamard type can be characterized in a similar way,
after taking account of the equivalences given by
the global phase, diagonal unitary transformations,
and changes of basis.

In this paper, we fixed the initial state $\varphi_\ast$.
It is also an important problem to study how the exact uniformity changes
when the initial state is varied,
and whether there exists a condition which gives the uniform distribution
in a form independent of the initial state.
By developing a classification for pairs of a coin matrix and an initial state,
we expect to obtain a clearer understanding of the essence
of the uniform distribution phenomenon for the Hadamard coin.

\section*{Conflict of interest}

The authors declare that they have no conflict of interest.

\section*{Data availability}

Data sharing is not applicable to this article as no datasets were generated or analysed during the current study.



\begin{thebibliography}{99}

\bibitem{Ambainis2001}
A.~Ambainis, E.~Bach, A.~Nayak, A.~Viswanath, and J.~Watrous,
\newblock One-dimensional quantum walks,
\newblock in {\em Proceedings of the 33rd Annual ACM Symposium on Theory of Computing (STOC)},
2001, pp.~37--49.

\bibitem{Konno2002}
N.~Konno,
\newblock Quantum random walks in one dimension,
\newblock {\em Quantum Information Processing} \textbf{1} (2002), 345--354.

\bibitem{Kempe2003}
J.~Kempe,
\newblock Quantum random walks: An introductory overview,
\newblock {\em Contemporary Physics} \textbf{44} (2003), 307--327.

\bibitem{VenegasAndraca2012}
S.~E.~Venegas-Andraca, 
\newblock Quantum walks: a comprehensive review,
\newblock {\em Quantum Information Processing} \textbf{11} (2012), 1015.

\bibitem{Konno2005}
N.~Konno,
\newblock A new type of limit theorems for the one-dimensional quantum random walk,
\newblock {\em Journal of the Mathematical Society of Japan} \textbf{57} (2005), 1179--1195.

\bibitem{Konno2012}
N.~Konno,
\newblock Sojourn times of the Hadamard walk in one dimension,
\newblock {\em Quantum Information Processing} \textbf{11} (2012), 465--480.

\bibitem{Cai2023}
S.~Cai, Q.~Huang, Y.~Ye, Y.~Wen, and Y.~Lin,
\newblock The sojourn times of one-dimensional discrete-time quantum walks,
\newblock {\em Laser Physics Letters} \textbf{20} (2023), 095210.

\end{thebibliography}
\end{document}